\documentclass[a4paper,12pt]{article}
\usepackage[utf8]{inputenc}

\usepackage{test}

\usepackage[authoryear,round]{natbib}

\usepackage{hyperref}
\hypersetup{colorlinks,citecolor=blue,urlcolor=magenta}
\usepackage{doi}

\usepackage[inline,shortlabels]{enumitem}
\setlist[enumerate,1]{label=(\roman*)}

\usepackage[margin = 1.25in]{geometry}
\usepackage{setspace}
\title{On Equilibrium Determinacy in Overlapping Generations Models with Money}
\author{Tomohiro Hirano\thanks{Department of Economics, Royal Holloway, University of London and Research Associate at the Center for Macroeconomics at the London School of Economics. Email: \href{mailto:tomohiro.hirano@rhul.ac.uk}{tomohih@gmail.com}.}
\and Alexis Akira Toda\thanks{Department of Economics, University of California San Diego. Email: \href{mailto:atoda@ucsd.edu}{atoda@ucsd.edu}.}}

\numberwithin{equation}{section}
\numberwithin{lem}{section}
\numberwithin{prop}{section}

\begin{document}

\maketitle

\begin{abstract}

This paper provides a detailed analysis of the local determinacy of monetary and non-monetary steady states in \citet{Tirole1985}'s classical two-period overlapping generations model with capital and production. We show that the sufficient condition for local determinacy in endowment economies provided by \citet{Scheinkman1980} does not generalize to models with production: there are robust examples with arbitrary utility functions in which the non-monetary steady state is locally determinate or indeterminate. In contrast, the monetary steady state is locally determinate under fairly weak conditions.

\medskip

\textbf{Keywords:} local determinacy, money, overlapping generations model, production.

\medskip

\textbf{JEL codes:} C62, D53, G12.
\end{abstract}

\section{Introduction}
It is well known that in overlapping generations (OLG) model with money, equilibrium indeterminacy is possible: there often exists a steady state in which money has value (monetary or bubbly steady state) as well as a continuum of equilibrium paths converging to the steady state in which money has no value (non-monetary or fundamental steady state). The possibility of equilibrium indeterminacy was first pointed out by \citet{Gale1973} in the endowment economy of \citet{Samuelson1958}. Somewhat surprisingly, formal analyses of the determinacy of steady states in these classical models are not easy to find.\footnote{\citet{GalorRyder1991} study the dynamic efficiency in \citet{Tirole1985}'s model without money, which is related to local determinacy of the non-monetary steady state. However, they only provide sufficient conditions using the constant elasticity of substitution (CES) production function. \citet{BoseRay1993} study the existence and efficiency of monetary equilibria in \citet{Tirole1985}'s model under high-level assumptions on the savings function. In contrast, we study local (in)determinacy in both non-monetary and monetary economies.} Many authors use simple phase diagrams or high-level assumptions to study the local determinacy of equilibria. For instance, \citet[p.~1502]{Tirole1985} directly assumes some monotonicity condition on a function that determines the rental rate without providing sufficient conditions on fundamentals. Similarly, \citet[p.~268, Endnote 16]{BlanchardFischer1989} state ``Care must be taken in using a phase diagram to analyze the dynamics of a \emph{difference} equation system. [\ldots] Thus we must check in this case whether the system is indeed saddle point stable [\ldots]. This check is left to the reader''.

In a by now forgotten contribution, \citet{Scheinkman1980} shows that if the old have no endowment and the period utility function satisfies $\lim_{x\to 0}xu'(x)>0$, then in endowment economies the non-monetary steady state is locally determinate. \citet{Santos1990} significantly generalizes this result and proves the existence of equilibria in which the value of money is bounded away from zero.

This paper revisits these classical models with money and provides a detailed analysis of equilibrium determinacy in \citet{Tirole1985}'s model. We obtain two main results, one negative and one positive. The first and negative result roughly states that, given an arbitrary utility function, by judiciously choosing the production function, we can provide robust examples of equilibrium determinacy or indeterminacy at the non-monetary steady state. This result shows that \citet{Scheinkman1980}'s sufficient condition does not generalize to production economies, and in fact, no conditions on utility function alone are sufficient for determinacy or indeterminacy. The second and positive result states that the monetary steady state (if one exists) is locally determinate under fairly weak conditions.

To illustrate our theoretical results, we provide a complete characterization for the model with Cobb-Douglas utility function and constant elasticity of substitution (CES) production function. Even in this canonical setting, we show that anything goes for the eigenvalues of the Jacobian at the non-monetary steady state by choosing parameters appropriately. However, we show that whenever a monetary steady state exists, it is always locally determinate, and there exists one locally indeterminate non-monetary steady state and potentially another that is unstable. Our paper highlights the importance of complementing intuition with formal mathematical analysis and shows that monetary economies are inherently unstable with a continuum of equilibria.

\section{Determinacy in OLG production economies}\label{sec:production}

\subsection{Model}

The model description is brief because it is well known. We consider \citet{Tirole1985}'s OLG model with capital accumulation and a constant population normalized to 1. Let $U(y,z)=u(y)+\beta v(z)$ be the utility when consumption is $(y,z)$ for young and old. We assume $u$ is twice continuously differentiable, $u'>0$, $u''<0$, satisfies the Inada condition $u'(0)=\infty$, and likewise for $v$. Let $F(K,L)$ be the aggregate neoclassical production function, where $K,L$ are capital and labor inputs. The young supply labor inelastically and invest in capital and an intrinsically useless asset in unit supply, which is initially held by the old.\footnote{\citet{ShiSuen2014} consider elastic labor supply.} The old liquidate all resources and consume.

Let $P_t$ be the asset price, $R_t$ the rental rate, $w_t$ the wage rate, $x_t$ the asset holdings of the young, $k_t$ the capital holdings of the old, $(y_t,z_t)$ the consumption of the young and old, and $(K_t,L_t)$ the capital and labor inputs. A competitive equilibrium consists of a sequence of these variables such that
\begin{enumerate*}
    \item each generation maximizes utility subject to the budget constraints $y_t+k_{t+1}+P_tx_t=w_t$ for the young and $z_{t+1}=R_{t+1}k_{t+1}+P_{t+1}x_t$ for the old,
    \item the firm maximizes profit $F(K_t,L_t)-R_tK_t-w_tL_t$,
    \item the commodity market clears, so $y_t+z_t+k_{t+1}=F(K_t,L_t)$,
    \item factor markets clear, so $K_t=k_t$, $L_t=1$,
    \item the asset market clears, so $x_t=1$.
\end{enumerate*}

It is easy to show that either $P_t=0$ for all $t$ or $P_t>0$ for all $t$ \citep[\S2]{HiranoToda2024JME}. We say that an equilibrium is
\begin{enumerate*}
    \item \emph{stationary} or a \emph{steady state} if $P_t=P$ is constant,
    \item  \emph{non-monetary} (or \emph{fundamental}) if $P_t=0$ for all $t$,
    \item \emph{monetary} (or \emph{bubbly}) if $P_t>0$ for all $t$.
\end{enumerate*}
In what follows, we abbreviate monetary and non-monetary steady states as MSS and NMSS, respectively.

\subsection{Equilibrium analysis}

Define $f(k)=F(k,1)$. Profit maximization implies
\begin{align*}
    R_t&=F_K(k_t,1)=f'(k_t),\\
    w_t&=F_L(k_t,1)=f(k_t)-k_tf'(k_t).
\end{align*}

It is convenient to define savings by $s_t\coloneqq k_{t+1}+P_tx_t=w_t-y_t$. In a non-monetary equilibrium, because $P_t=0$, we have $s_t=k_{t+1}$. In a monetary equilibrium, because capital and the asset are perfect substitutes, agents must be indifferent between holding the two. In either case, the young's problem reduces to maximizing
\begin{equation*}
    u(w_t-s_t)+\beta v(R_{t+1}s_t).
\end{equation*}
The strict concavity of $u,v$ and the Inada condition imply that there exists a unique solution $s_t=s(w_t,R_{t+1})$, where $s(w,R)$ solves the first-order condition
\begin{equation}
    -u'(w-s)+\beta Rv'(Rs)=0. \label{eq:foc}
\end{equation}

In equilibrium, because $x_t=1$, we obtain the equilibrium condition
\begin{equation}
    k_{t+1}+P_t=s(f(k_t)-k_tf'(k_t),f'(k_{t+1})). \label{eq:eqcond1}
\end{equation}
Furthermore, individual optimality implies the no-arbitrage condition
\begin{equation}
    P_t=\frac{1}{f'(k_{t+1})}P_{t+1}. \label{eq:noarbitrage}
\end{equation}
Therefore the equilibrium is characterized by the system of nonlinear difference equations \eqref{eq:eqcond1} and \eqref{eq:noarbitrage}. To study the equilibrium dynamics, define the state variable $\xi_t=(\xi_{1t},\xi_{2t})=(k_t,P_t)$. Then we may write \eqref{eq:eqcond1}, \eqref{eq:noarbitrage} as $\Phi(\xi_t,\xi_{t+1})=0$, where
\begin{align*}
    \Phi_1(\xi,\eta)&=\eta_1+\xi_2-s(f(\xi_1)-\xi_1f'(\xi_1),f'(\eta_1)),\\
    \Phi_2(\xi,\eta)&=\eta_2-\xi_2f'(\eta_1).
\end{align*}
The following lemma is a simple application of the implicit function theorem. Below, we denote the Jacobian of a function $g$ with respect to a variable $x$ by $D_xg$.

\begin{lem}\label{lem:DPhi}
Let $\xi^*=(k,P)$ be a steady state, so $\Phi(\xi^*,\xi^*)=0$. If $D_\eta\Phi(\xi^*,\xi^*)$ is nonsingular, the equation $\Phi(\xi,\eta)=0$ can be locally solved as $\eta=\phi(\xi)$, where
\begin{equation}
    D_\xi\phi(\xi^*)=\frac{1}{1-s_Rf''}\begin{bmatrix}
        -s_wkf'' & -1\\
        -s_wkP(f'')^2 & (1-s_Rf'')f'-Pf''
    \end{bmatrix} \label{eq:Dphi_ss}
\end{equation}
and $f',f''$ are evaluated at $k$.
\end{lem}

At NMSS, setting $P=0$ in \eqref{eq:Dphi_ss} yields
\begin{equation}
    D_\xi\phi=\begin{bmatrix}
        -\frac{s_wkf''}{1-s_Rf''} & -\frac{1}{1-s_Rf''}\\
        0 & f'
    \end{bmatrix}. \label{eq:Dphi_f}
\end{equation}
Studying local determinacy reduces to determining whether the diagonal entries of $D_\xi\phi$ in \eqref{eq:Dphi_f} are inside or outside the unit circle. Without further structure, it is difficult to provide general sufficient conditions. We thus approach the problem in reverse, and ask whether we can reverse-engineer a model that generates a desired steady state. The following lemma provides an answer.

\begin{lem}
\label{lem:reverse}
Take any utility functions of the young and old $u,v$, capital $k^*>0$, rental rate $R>0$, wage $w>k^*$, and parameter $c\le 0$. Then there exist discount factor $\beta>0$ and production function $f$ with $f''(k^*)=c$ such that $(k^*,0)$ is NMSS consistent with the given parameters.
\end{lem}

We can show that ``anything goes'' regarding the local determinacy of NMSS, which is our first main result.

\begin{thm}
\label{thm:locdet_f}
Let everything be as in Lemma \ref{lem:reverse} and $c\le 0$ sufficiently close to 0. Then the following statements are true.
\begin{enumerate}
    \item If $R<1$, then NMSS is locally indeterminate: for any $(k_0,P_0)$ sufficiently close to $(k^*,0)$, there exists a monetary equilibrium converging to NMSS.
    \item If $R>1$, then NMSS is locally determinate: for any $k_0$ sufficiently close to $k^*$, there exists a unique equilibrium converging to NMSS, which is non-monetary.
\end{enumerate}
\end{thm}

Theorem \ref{thm:locdet_f} implies that no condition on the utility functions $u,v$ \emph{alone} are sufficient for the local determinacy of NMSS in a production economy. This is very different from endowment economies discussed in Appendix \ref{sec:endowment}, for which there exists a condition on $v$ alone that is sufficient for local determinacy.

Our second main result shows that there exists a condition on $v$ alone that guarantees the local determinacy of MSS.

\begin{thm}
\label{thm:locdet_b}
Let $(k,P)\gg 0$ be MSS. If $f''<0$ and $1-s_Rf''>0$ at the steady state, then $(k,P)$ is locally determinate. In particular, a sufficient condition for local determinacy is $\gamma_v\le 1$, where $\gamma_v(z)\coloneqq -zv''(z)/v'(z)>0$ denotes the relative risk aversion of $v$.
\end{thm}

\section{Cobb-Douglas-CES economy}

To illustrate Theorems \ref{thm:locdet_f} and \ref{thm:locdet_b}, we present a canonical example. Let the utility function be Cobb-Douglas, so
\begin{equation*}
    U(y,z)=u(y)+\beta v(z)=(1-\beta)\log y+\beta \log z
\end{equation*}
with $\beta\in (0,1)$. Let the production function exhibit constant elasticity of substitution (CES), so
\begin{equation*}
    F(K,L)=\begin{cases*}
        A\left(\alpha K^{1-\rho}+(1-\alpha)L^{1-\rho}\right)^\frac{1}{1-\rho}+(1-\delta)K & if $0<\rho\neq 1$,\\
        AK^\alpha L^{1-\alpha}+(1-\delta)K & if $\rho=1$,
    \end{cases*}
\end{equation*}
where $A>0$ is productivity, $\alpha\in (0,1)$, $\delta\in (0,1]$ is the capital depreciation rate, and $1/\rho$ is the capital-labor substitution elasticity. Below, for brevity we refer to this model as $\mathcal{M}(\theta)$, where $\theta=(\beta,A,\alpha,\rho,\delta)$ lists model parameters.

With Cobb-Douglas utility, the savings function is $s(w,R)=\beta w$, implying $s_w=\beta$ and $s_R=0$. Therefore the eigenvalues of the Jacobian at NMSS \eqref{eq:Dphi_f} are $\lambda_1\coloneqq -\beta kf''(k)>0$ and $\lambda_2\coloneqq f'(k)>1-\delta$. The following proposition shows that we can reverse-engineer $\mathcal{M}(\theta)$ that achieves any admissible $\lambda_1,\lambda_2$.

\begin{prop}\label{prop:lambda}
Let $\beta\in (0,1)$, $\delta\in (0,1]$, $k>0$, $\lambda_1>0$, and $\lambda_2>1-\delta$ be given. Then there exist unique $A>0$, $\alpha\in (0,1)$, and $\rho>0$ such that the non-monetary steady state capital of $\mathcal{M}(\theta)$ is $k$ and the eigenvalues of the Jacobian are $\lambda_1,\lambda_2$.
\end{prop}

he value added of Proposition \ref{prop:lambda} relative to Theorem \ref{thm:locdet_f} is that even in the canonical setting of Cobb-Douglas-CES economy, anything goes for local determinacy of NMSS. However, Proposition \ref{prop:lambda} is silent about the uniqueness of NMSS or the existence of MSS for a given $\mathcal{M}(\theta)$, which we study next. The following proposition provides a necessary and sufficient condition for the existence, uniqueness, and local determinacy of MSS.

\begin{prop}\label{prop:monetary}
In any $\mathcal{M}(\theta)$, MSS exists if and only if $\beta\delta(c-1)>1$, where $c\coloneqq \alpha^{-\frac{1}{\rho}}(\delta/A)^\frac{1-\rho}{\rho}$. MSS is unique and locally determinate whenever it exists.
\end{prop}

The following proposition characterizes the existence, uniqueness, and local determinacy of NMSS for the case $\rho\le 1$.

\begin{prop}\label{prop:rho<1}
In any $\mathcal{M}(\theta)$ with $\rho\le 1$, there exists a unique NMSS. Furthermore, NMSS is locally indeterminate ($\lambda_1,\lambda_2\in (0,1)$) if and only if MSS exists.
\end{prop}

Proposition \ref{prop:rho<1} implies that with a CES production function with elasticity of substitution $1/\rho\ge 1$ (including Cobb-Douglas), whenever MSS exists, NMSS is locally indeterminate. This is consistent with the textbook phase diagram of \citet[\S5.2]{BlanchardFischer1989}. However, empirical evidence suggests $1/\rho<1$ and hence $\rho>1$ \citep*{OberfieldRaval2021,GechertHavranekIrsovaKolcunova2022}. The following theorem provides a characterization for the case $\rho>1$.

\begin{prop}\label{prop:rho>1}
In any $\mathcal{M}(\theta)$ with $\rho>1$, the following statements are true.
\begin{enumerate}
    \item\label{item:rho1} NMSS exists if and only if $\rho^\rho\alpha [\beta A(\rho-1)]^{1-\rho}\le 1$. The number of NMSS is 2 if the inequality is strict and 1 if equality holds. 
    \item\label{item:rho2} Whenever MSS exists, there exist exactly two NMSS, one of which is locally indeterminate ($\lambda_1,\lambda_2\in (0,1)$) and the other unstable ($\lambda_1,\lambda_2>1$).
\end{enumerate}
\end{prop}

If $\rho>1$, the conditions in Propositions \ref{prop:monetary} and \ref{prop:rho>1}\ref{item:rho1} simultaneously hold (fail) if $A>0$ is sufficiently large (small). Therefore the coexistence of MSS and NMSS as well as the nonexistence of any steady states are both possible. As discussed in the proof of Proposition \ref{prop:rho>1}, when condition \ref{item:rho1} holds, the first eigenvalues of the two NMSS satisfy $\lambda_1\le 1\le \lambda_1'$ (with equality if and only if NMSS is unique). However, the sign of $\lambda_2-1$ depends on other parameters when MSS does not exist.

%\bibliographystyle{plainnat}
%\bibliography{localbib}

\appendix

\section{Proofs}

\begin{proof}[Proof of Lemma \ref{lem:DPhi}]
A straightforward calculation yields
\begin{align*}
    D_\xi\Phi&=\begin{bmatrix}
        s_w\xi_1f''(\xi_1) & 1\\
        0 & -f'(\eta_1)
    \end{bmatrix},\\
    D_\eta\Phi&=\begin{bmatrix}
        1-s_Rf''(\eta_1) & 0\\
        -\xi_2f''(\eta_1) & 1
    \end{bmatrix}.
\end{align*}
If $D_\eta\Phi$ is nonsingular, we may apply the implicit function theorem to obtain
\begin{align*}
    D_\xi\phi&=-[D_\eta\Phi]^{-1}D_\xi\Phi\\
    &=\frac{1}{1-s_Rf''(\eta_1)}\begin{bmatrix}
        -s_w\xi_1f''(\xi_1) & -1\\
        -s_w\xi_1\xi_2f''(\xi_1)f''(\eta_1) & (1-s_Rf''(\eta_1))f'(\eta_1)-\xi_2f''(\eta_1)
    \end{bmatrix}.
\end{align*}
Setting $\xi=\eta=(k,P)$ yields \eqref{eq:Dphi_ss}.
\end{proof}

\begin{proof}[Proof of Lemma \ref{lem:reverse}]
NMSS conditions are utility maximization \eqref{eq:foc}, profit maximization $R=f'(k^*)$ and $w=f(k^*)-k^*f'(k^*)$, and market clearing condition $k^*=s$. Let us construct such a model with $f''(k^*)=c$. First, given $u$, $v$, $k^*>0$, $R>0$, and $w>k^*$, set $s=k^*$ and take a discount factor $\beta>0$ such that $u'(w-s)=\beta Rv'(Rs)$. Next, define the output $Y\coloneqq w+Rk^*$. Since $w>0$, we can take a nonnegative, increasing, and concave function $f$ such that $f(k^*)=Y$, $f'(k^*)=R$, and $f''(k^*)=c$. For instance, for small enough $\epsilon>0$, we can let
\begin{equation*}
    f(k)=Y+R(k-k^*)+\frac{c}{2}(k-k^*)^2=w+Rk+\frac{c}{2}(k-k^*)^2
\end{equation*}
for $k\in [k^*-\epsilon,k^*+\epsilon]$ and linearly extrapolate outside that range. (The condition $w>0$ ensures that this $f$ is nonnegative as long as $\epsilon>0$ is small enough.) With this $f$, we have $R=f'(k^*)$ and $w=Y-Rk^*=f(k^*)-k^*f'(k^*)$, so the first-order conditions for profit maximization are satisfied. Therefore $(k^*,0)$ is indeed NMSS.
\end{proof}

\begin{proof}[Proof of Theorem \ref{thm:locdet_f}]
For simplicity, suppose $c=0$. (The proof for $c<0$ sufficiently close to 0 is similar.) Since $f'(k^*)=R$ and $f''(k^*)=c=0$, the Jacobian \eqref{eq:Dphi_f} reduces to
\begin{equation*}
    D_\xi\phi=\begin{bmatrix}
        0 & -1 \\
        0 & R
    \end{bmatrix},
\end{equation*}
whose eigenvalues are $0,R$. If $R<1$, the steady state $(k^*,0)$ is stable, and therefore for any $(k_0,P_0)$ sufficiently close to $(k^*,0)$, there exists an equilibrium path converging to the steady state.

If $R>1$, the steady state $(k^*,0)$ is a saddle point. Therefore for any $k_0$ sufficiently close to $k^*$, there exists a unique equilibrium path converging to the steady state. On such a path, because $\set{P_t}$ is bounded and $R>1$, we have $\lim_{t\to\infty}R^{-t}P_t=0$. Since the transversality condition for asset pricing holds (see \citet[\S2]{HiranoToda2024JME}), the asset price $P_t$ must equal its fundamental value, so $P_t=0$ for all $t$.
\end{proof}

\begin{proof}[Proof of Theorem \ref{thm:locdet_b}]
A simple application of the implicit function theorem shows
\begin{subequations}\label{eq:Ds}
\begin{align}
    s_w(w,R)&=\frac{u''(w-s)}{u''(w-s)+\beta R^2v''(Rs)}\in (0,1), \label{eq:sw}\\
    s_R(w,R)&=-\frac{\beta v'(Rs)+\beta Rsv''(Rs)}{u''(w-s)+\beta R^2v''(Rs)}. \label{eq:sR}
\end{align}
\end{subequations}

In MSS, the no-arbitrage condition \eqref{eq:noarbitrage} implies $f'(k)=1$. Therefore letting $J$ be the Jacobian \eqref{eq:Dphi_ss} and $c=f''<0$, we have
\begin{align*}
    t&\coloneqq\tr J=1-\frac{(s_wk+P)c}{1-s_Rc},\\
    d&\coloneqq\det J=-\frac{s_wkc}{1-s_Rc}.
\end{align*}
Let the characteristic polynomial of $J$ be
\begin{equation*}
    p(x)\coloneqq \det(xI-J)=x^2-tx+d.
\end{equation*}
By assumption, $1-s_Rc>0$. Since $s_w>0$ by \eqref{eq:sw} and $c<0$, it follows that $p(0)=d>0$. Furthermore,
\begin{equation*}
    p(1)=1-t+d=\frac{Pc}{1-s_Rc}<0.
\end{equation*}
Therefore the two roots $\lambda_1,\lambda_2$ of $p$ satisfy $0<\lambda_1<1<\lambda_2$. Since the steady state is a saddle point, MSS is locally determinate.

Since $c<0$, the condition $1-s_Rc>0$ trivially holds if $s_R\ge 0$. Noting that $u''<0$, $v''<0$, and $v'(z)+zv''(z)=v'(z)(1-\gamma_v(z))$, by \eqref{eq:sR} a sufficient condition for $s_R\ge 0$ is $\gamma_v\le 1$.
\end{proof}

\begin{proof}[Proof of Proposition \ref{prop:lambda}]
We focus on the case $\rho\neq 1$, though the Cobb-Douglas case with $\rho=1$ can be analyzed by taking the limit $\rho\to 1$.

Define $g(k)=(\alpha k^{1-\rho}+1-\alpha)^\frac{1}{1-\rho}$. A straightforward calculation yields
\begin{align*}
    g'&=\alpha k^{-\rho}g^\rho,\\
    g''&=-\rho\alpha(1-\alpha)k^{-\rho-1}g^{2\rho-1}.
\end{align*}
The equilibrium condition \eqref{eq:eqcond1} at NMSS is
\begin{equation}
    k=\beta w=\beta A(1-\alpha)g^\rho\iff A=\frac{1}{\beta(1-\alpha)}kg^{-\rho}. \label{eq:A}
\end{equation}
Therefore $A>0$ is uniquely determined once we determine $\alpha,\rho$. At NMSS, we have
\begin{subequations}\label{eq:lambda}
\begin{align}
    \lambda_1&=-\beta kf''(k)=\beta k A \rho\alpha(1-\alpha)k^{-\rho-1}g^{2\rho-1} \notag \\
    &=\rho\alpha k^{1-\rho}g^{\rho-1}=\rho\frac{\alpha k^{1-\rho}}{\alpha k^{1-\rho}+1-\alpha}, \label{eq:lambda1}\\
    \lambda_2&=f'(k)=A\alpha k^{-\rho}g^\rho+1-\delta \notag\\
    &=\frac{\alpha}{\beta(1-\alpha)}k^{1-\rho}+1-\delta. \label{eq:lambda2}
\end{align}
\end{subequations}
Solving \eqref{eq:lambda2} for $\alpha k^{1-\rho}$ and substituting into \eqref{eq:lambda1}, we obtain
\begin{equation}
    \lambda_1=\rho\frac{\beta(\lambda_2-1+\delta)}{\beta(\lambda_2-1+\delta)+1}\iff \rho=\lambda_1\frac{\beta(\lambda_2-1+\delta)+1}{\beta(\lambda_2-1+\delta)}. \label{eq:rho}
\end{equation}
Therefore $\rho$ is uniquely determined as in \eqref{eq:rho}. Because the right-hand side of \eqref{eq:lambda2} is strictly increasing in $\alpha\in (0,1)$ and its range is $(1-\delta,\infty)$, there exists a unique $\alpha\in (0,1)$ satisfying \eqref{eq:lambda2}.
\end{proof}

We prove Propositions \ref{prop:monetary}--\ref{prop:rho>1} by establishing a series of lemmas.

\begin{lem}\label{lem:rho=1}
The conclusions of Propositions \ref{prop:monetary}, \ref{prop:rho<1} hold if $\rho=1$.
\end{lem}

\begin{proof}
If $\rho=1$, the equilibrium condition \eqref{eq:eqcond1} at NMSS is
\begin{equation*}
    k=s=\beta A(1-\alpha)k^\alpha\iff k_f\coloneqq [\beta A(1-\alpha)]^\frac{1}{1-\alpha},
\end{equation*}
which uniquely exists. At this $k_f$, \eqref{eq:lambda} implies $\lambda_1=\alpha<1$ and $\lambda_2=\frac{\alpha}{\beta(1-\alpha)}+1-\delta$. Therefore NMSS is locally indeterminate if and only if $\lambda_2<1$, which is equivalent to $\beta\delta(c-1)>1$ because $c=1/\alpha$. If MSS exists, we have
\begin{equation*}
    1=f'(k)=A\alpha k^{\alpha-1}+1-\delta\iff k_b\coloneqq (A\alpha/\delta)^\frac{1}{1-\alpha}.
\end{equation*}
The equilibrium condition \eqref{eq:eqcond1} at $k=k_b$ implies
\begin{equation*}
    P=s-k=\beta A(1-\alpha)k^\alpha-k=(A\alpha/\delta)^\frac{1}{1-\alpha}\left(\beta\delta\frac{1-\alpha}{\alpha}-1\right).
\end{equation*}
Therefore a necessary and sufficient condition for the existence (and uniqueness) of MSS is $P>0$, or equivalently, $\beta\delta(c-1)>1$.
\end{proof}

\begin{lem}\label{lem:exist_b}
Let $\rho\neq 1$ and $c\coloneqq \alpha^{-\frac{1}{\rho}}(\delta/A)^\frac{1-\rho}{\rho}$. Then MSS exists if and only if $\beta\delta(c-1)>1$, in which case the unique MSS capital is given by $k=\left(\frac{1/\alpha-1}{c-1}\right)^\frac{1}{1-\rho}$.
\end{lem}

\begin{proof}
The uniqueness of $k$ follows from the strict concavity of $f$.

If MSS exists, using \eqref{eq:lambda2}, we have
\begin{equation*}
    1=f'(k)=A\alpha(g/k)^\rho+1-\delta\iff g/k=(A\alpha/\delta)^{-1/\rho}.
\end{equation*}
Define $x=\alpha k^{1-\rho}$. Then this condition is equivalent to
\begin{equation}
    \frac{x+1-\alpha}{x/\alpha}=(A\alpha/\delta)^\frac{\rho-1}{\rho}\iff x=\frac{1-\alpha}{\alpha^{-\frac{1}{\rho}}(\delta/A)^\frac{1-\rho}{\rho}-1}=\frac{1-\alpha}{c-1}, \label{eq:xb}
\end{equation}
so $c>1$ is necessary. Under this condition, the steady state capital is $k=(x/\alpha)^\frac{1}{1-\rho}=\left(\frac{1/\alpha-1}{c-1}\right)^\frac{1}{1-\rho}$.

For this $k$ to be MSS, it is necessary and sufficient that $P=\beta w-k>0$, where $w$ is the wage. Using \eqref{eq:A}, we obtain
\begin{equation}
    \beta w/k=\beta A(1-\alpha)g^\rho/k=\beta A(1-\alpha)\alpha^\frac{1}{1-\rho}\psi(x)^\frac{1}{1-\rho}, \label{eq:P>0}
\end{equation}
where $\psi(x)\coloneqq (x+1-\alpha)^\rho/x$. Using \eqref{eq:xb}, we obtain
\begin{align*}
    & x+1-\alpha=\frac{c(1-\alpha)}{c-1}=cx \\
    \iff & \psi(x)=\frac{(x+1-\alpha)^\rho}{x}=c^\rho x^{\rho-1}\\
    \iff & \psi(x)^\frac{1}{1-\rho}=c^\frac{\rho}{1-\rho}\frac{c-1}{1-\alpha}=\alpha^{-\frac{1}{1-\rho}}(\delta/A)\frac{c-1}{1-\alpha}.
\end{align*}
Therefore by \eqref{eq:P>0}, we obtain
\begin{align*}
    P>0 &\iff \beta A(1-\alpha)\alpha^\frac{1}{1-\rho}\alpha^{-\frac{1}{1-\rho}}(\delta/A)\frac{c-1}{1-\alpha}>1\\
    &\iff \beta\delta(c-1)>1. \qedhere
\end{align*}
\end{proof}

\begin{proof}[Proof of Proposition \ref{prop:monetary}]
The existence and uniqueness of MSS are immediate from Lemmas \ref{lem:rho=1} and \ref{lem:exist_b}. Since $\gamma_v(z)=1$, Theorem \ref{thm:locdet_b} implies that MSS is locally determinate.
\end{proof}

\begin{lem}\label{lem:exist_f}
If $\rho<1$, a unique NMSS exists. If $\rho>1$, NMSS exists if and only if $\rho^\rho\alpha [\beta A(\rho-1)]^{1-\rho}\le 1$. The number of NMSS is 2 if the inequality is strict and 1 if equality holds.
\end{lem}

\begin{proof}
We use the same notation as in the proof of Lemma \ref{lem:exist_b}. By \eqref{eq:P>0}, NMSS $x$ satisfies
\begin{equation*}
    \beta A(1-\alpha)\alpha^\frac{1}{1-\rho}\psi(x)^\frac{1}{1-\rho}\iff \psi(x)=\frac{1}{\alpha[\beta A(1-\alpha)]^{1-\rho}}.
\end{equation*}
A straightforward calculation yields
\begin{equation}
    \psi'(x)=\frac{(x+1-\alpha)^{\rho-1}}{x^2}((\rho-1)x-(1-\alpha)). \label{eq:psi'}
\end{equation}
If $\rho<1$, then $\psi'(x)<0$, $\psi(0)=\infty$, and $\psi(\infty)=0$, so a unique NMSS exists. If $\rho>1$, then $\psi$ achieves a unique minimum at $m\coloneqq \frac{1-\alpha}{\rho-1}$ and $\psi(0)=\psi(\infty)=\infty$. Therefore NMSS exists if and only if
\begin{equation*}
    \rho^\rho\left(\frac{1-\alpha}{\rho-1}\right)^{\rho-1}=\psi(m)\le \psi(x)=\frac{1}{\alpha[\beta A(1-\alpha)]^{1-\rho}},
\end{equation*}
which is equivalent to the desired condition. Furthermore, since $\psi'(x)\gtrless 0$ according as $x\lessgtr m$, we obtain the claim.
\end{proof}

\begin{lem}\label{lem:rho<1}
If $\rho<1$, then MSS exists if and only if $\lambda_2<1$ at NMSS.
\end{lem}

\begin{proof}
We use the same notation as in the proof of Lemma \ref{lem:exist_b}. Let $x_b,x_f$ and $k_b,k_f$ be the $x$ and $k$ in the (necessarily unique) MSS and NMSS, where $x=\alpha k^{1-\rho}$. Then $\rho<1$, \eqref{eq:P>0}, the monotonicity of $\psi$, and the strict concavity of $f$ imply that MSS exists if and only if
\begin{align}
    & \beta A(1-\alpha)\alpha^\frac{1}{1-\rho}\psi(x_b)^\frac{1}{1-\rho}>1=\beta A(1-\alpha)\alpha^\frac{1}{1-\rho}\psi(x_f)^\frac{1}{1-\rho} \label{eq:psixbf}\\
    \iff & \psi(x_b)>\psi(x_f)\iff x_b<x_f \notag \\
    \iff & k_b<k_f \iff 1=f'(k_b)>f'(k_f)=\lambda_2. \notag \qedhere
\end{align}
\end{proof}

\begin{proof}[Proof of Proposition \ref{prop:rho<1}]
The case $\rho=1$ follows from Lemma \ref{lem:rho=1}. If $\rho<1$, then \eqref{eq:lambda1} implies $\lambda_1\in (0,1)$ at NMSS. The conclusion holds by Lemmas \ref{lem:exist_b} and \ref{lem:rho<1}.
\end{proof}

\begin{proof}[Proof of Proposition \ref{prop:rho>1}]
\ref{item:rho1} Obvious from Lemma \ref{lem:exist_f}.

\ref{item:rho2} Solving for $\alpha A^{1-\rho}$, the existence conditions for MSS in Proposition \ref{prop:monetary} and for NMSS in Lemma \ref{lem:exist_f} are, respectively,
\begin{subequations}
\begin{align}
    \alpha A^{1-\rho}&<\delta(1/\beta+\delta)^{-\rho}, \label{eq:cond_b}\\
    \alpha A^{1-\rho}&\le \rho^{-\rho}[\beta(\rho-1)]^{\rho-1}. \label{eq:cond_f}
\end{align}
\end{subequations}
Using calculus, it is straightforward to show that the maximum of the right-hand side of \eqref{eq:cond_b} over $\delta>0$ is achieved at $\delta=\frac{1}{\beta(\rho-1)}$, in which case the maximum value is equal to the right-hand side of \eqref{eq:cond_f}. Therefore whenever MSS exists, there exist exactly two NMSS.

Suppose MSS $x_b$ exists and denote the two NMSS by $x_f<x_f'$. Then $\rho>1$ and \eqref{eq:psixbf} imply $\psi(x_b)<\psi(x_f)=\psi(x_f')$. Since $\psi$ is strictly quasi-concave, it must be $x_f<x_b<x_f'$. Since $x=\alpha k^{1-\rho}$ and $\rho>1$, we have $k_f>k_b>k_f'$. Since $f$ is strictly concave, we have
\begin{equation*}
    0<\lambda_2=f'(k_f)<f'(k_b)=1<f'(k_f')=\lambda_2'.
\end{equation*}
Again, the strict quasi-concavity of $\psi$ implies $x_f<m<x_f'$, where $m=\frac{1-\alpha}{\rho-1}$. Therefore using \eqref{eq:lambda1}, we obtain
\begin{equation*}
    0<\lambda_1=\rho\frac{x_f}{x_f+1-\alpha}<\rho\frac{m}{m+1-\alpha}=1<\rho\frac{x_f'}{x_f'+1-\alpha}=\lambda_1'.
\end{equation*}
Note this last inequality is independent of whether $x_b$ exists or not.
\end{proof}

\section{Determinacy in OLG endowment economies}\label{sec:endowment}

This appendix reviews the local determinacy of monetary and NMSS in the classical two-period OLG endowment economies.

\subsection{Model}

Consider the classical two-period OLG model with a single perishable good. Preferences are the same as in \S\ref{sec:production}. Each period, the young and old have endowments $(a,b)$, where $a>0$ and $b\ge 0$. In addition, the initial old are endowed with an intrinsically useless asset (\ie, an asset paying no dividends like fiat money) in unit supply.

Letting $P_t\ge 0$ be the price of the asset at time $t$ in units of the consumption good, the budget constraints of generation $t$ are
\begin{align*}
    &\text{Young}: & y_t+P_tx_t&=a,\\
    &\text{Old}: & z_{t+1}&=b+P_{t+1}x_t,
\end{align*}
where $x_t$ is asset holdings. As usual, a competitive equilibrium is defined by a sequence $\set{(P_t,x_t,y_t,z_t)}_{t=0}^\infty$ such that
\begin{enumerate*}
    \item each generation maximizes utility subject to the budget constraints,
    \item commodity markets clear, so $y_t+z_t=a+b$,
    \item asset markets clear, so $x_t=1$.
\end{enumerate*}

In any equilibrium, the budget constraints imply the equilibrium allocation $(y_t,z_{t+1})=(a-P_t,b+P_{t+1})$. Therefore the equilibrium is completely determined by the price sequence $\set{P_t}_{t=0}^\infty$.
\iffalse
It is easy to show that either $P_t=0$ for all $t$ or $P_t>0$ for all $t$ \citep[\S2]{HiranoToda2024JME}. We introduce some terminology to classify equilibria. We say that an equilibrium is
\begin{enumerate*}
    \item \emph{stationary} or a \emph{steady state} if $P_t=P$ is constant,
    \item  \emph{non-monetary} (or \emph{fundamental}) if $P_t=0$ for all $t$,
    \item \emph{monetary} (or \emph{bubbly}) if $P_t>0$ for all $t$,
    \item \emph{asymptotically non-monetary} if $\liminf_{t\to\infty}P_t=0$,
    \item \emph{asymptotically monetary} if $\liminf_{t\to\infty}P_t>0$.
\end{enumerate*}
By definition, non-monetary equilibria are asymptotically non-monetary, and asymptotically monetary equilibria are monetary. However, there could be monetary equilibria that are asymptotically non-monetary.\footnote{An intrinsically useless asset is often called \emph{money} or \emph{pure bubble}. Strictly speaking, a bubble is a situation in which the asset price exceeds its fundamental value defined by the present value of dividends. For bubbles attached to real assets yielding positive dividends such as land, housing, and stocks, see \citet*{HiranoJinnaiTodaLeverage,HiranoTodaNecessity,HiranoTodaUnbalanced}.}
\fi

\subsection{Equilibrium analysis}
Because the non-monetary equilibrium is obviously unique ($P_t\equiv 0$), we focus on monetary equilibria. Take any monetary equilibrium, so $P_t>0$ for all $t$. Using the budget constraints to eliminate $(y_t,z_{t+1})$, generation $t$ seeks to maximize
\begin{equation*}
    u(a-P_tx_t)+\beta v(b+P_{t+1}x_t).
\end{equation*}
Taking the first-order condition and imposing the market clearing condition $x_t=1$, we obtain the equilibrium condition
\begin{equation}
    -u'(a-P_t)P_t+\beta v'(b+P_{t+1})P_{t+1}=0. \label{eq:eqcond}
\end{equation}
The following lemma provides a necessary and sufficient condition for the existence of MSS.

\begin{lem}\label{lem:ss_bubbly}
There exists MSS if and only if $u'(a)<\beta v'(b)$.
\end{lem}

\begin{proof}
If MSS $P>0$ exists, setting $P_t=P_{t+1}=P$ in \eqref{eq:eqcond} and dividing by $P>0$, we obtain
\begin{equation*}
    \psi(P)\coloneqq -u'(a-P)+\beta v'(b+P)=0.
\end{equation*}
Since $\psi'(P)=u''(a-P)+\beta v''(b+P)<0$, it follows that
\begin{equation*}
    0=\psi(P)<\psi(0)=-u'(a)+\beta v'(b)\implies u'(a)<\beta v'(b).
\end{equation*}

Conversely, suppose $u'(a)<\beta v'(b)$. Then
\begin{equation*}
    \psi(0)=-u'(a)+\beta v'(b)>0>-\infty=-u'(0)+\beta v'(b+a)=\psi(a).
\end{equation*}
By the intermediate value theorem, there exists $P>0$ such that $\psi(P)=0$. Furthermore, since $f$ is strictly decreasing, such a $P$ is unique. Therefore there exists a unique MSS.
\end{proof}

We next study the local determinacy of equilibria around the steady states. To this end, we write the equilibrium condition \eqref{eq:eqcond} as $\Phi(P_t,P_{t+1})=0$, where
\begin{equation}
    \Phi(\xi,\eta)=-u'(a-\xi)\xi+\beta v'(b+\eta)\eta. \label{eq:Phi}
\end{equation}
To apply the implicit function theorem, we compute the partial derivatives
\begin{subequations}\label{eq:Phi_partial}
\begin{align}
    \Phi_\xi(\xi,\eta)&=-u'(a-\xi)+u''(a-\xi)\xi, \label{eq:Phi_xi}\\
    \Phi_\eta(\xi,\eta)&=\beta v'(b+\eta)+\beta v''(b+\eta)\eta. \label{eq:Phi_eta}
\end{align}
\end{subequations}
Therefore if $\Phi_\eta\neq 0$, we may locally solve $\Phi(\xi,\eta)=0$ as $\eta=\phi(\xi)$, where
\begin{equation}
    \phi'(\xi)=-\frac{\Phi_\xi}{\Phi_\eta}=\frac{u'(a-\xi)-u''(a-\xi)\xi}{\beta v'(b+\eta)+\beta v''(b+\eta)\eta}. \label{eq:phi'}
\end{equation}
Let $P$ be a steady state. If $\abs{\phi'(P)}\neq 1$, the Hartman-Grobman theorem \citep[Theorem 4.6]{Chicone2006} implies that the local behavior of the dynamical system $\Phi(P_t,P_{t+1})=0$ (or $P_{t+1}=\phi(P_t)$) is qualitatively the same as that of the linearized system
\begin{equation*}
    P_{t+1}-P=\phi'(P)(P_t-P).
\end{equation*}
Therefore if $\abs{\phi'(P)}>1$, then the unique $\set{P_t}$ converging to $P$ is $P_t=P$ and the equilibrium is \emph{locally determinate}; if $\abs{\phi'(P)}<1$, then for any $P_0$ sufficiently close to $P$, there exists $\set{P_t}$ converging to $P$ and the equilibrium is \emph{locally indeterminate}.

The following proposition characterizes the local determinacy of NMSS.

\begin{prop}[\citealp{Gale1973}, Theorem 4]\label{prop:b>0}
Let $b>0$.
\begin{enumerate}
    \item If $u'(a)>\beta v'(b)$, NMSS is locally determinate, \ie, if $\set{P_t}_{t=0}^\infty$ is an equilibrium with $P_t\to 0$, then $P_t=0$ for all $t$.
    \item If $u'(a)<\beta v'(b)$, NMSS is locally indeterminate, \ie, for any $P_0>0$ sufficiently close to 0, there exists a monetary equilibrium $\set{P_t}_{t=0}^\infty$ with $P_t\to 0$.
\end{enumerate}
\end{prop}

\begin{proof}
Setting $\xi=\eta=P=0$ in \eqref{eq:phi'}, we obtain $\phi'(0)=u'(a)/\beta v'(b)>0$. The claim follows from the definition of local (in)determinacy.
\end{proof}

Combining Lemma \ref{lem:ss_bubbly} and Proposition \ref{prop:b>0}, we conclude that whenever MSS exists, there exists an asymptotically monetary equilibrium as well as a continuum of monetary but asymptotically non-monetary equilibria. \citet{Wallace1980} refers to this fact as the \emph{tenuousness} of monetary equilibria.

We next study the local determinacy of MSS. The following proposition shows that MSS is locally determinate whenever the relative risk aversion of the old is sufficiently low.\footnote{When the relative risk aversion of the old is sufficiently high, it is well known that cyclic and chaotic equilibrium dynamics are possible in monetary economies \citep{Grandmont1985}. See \citet{BoldrinWoodford1990} for a review of this type of models.}

\begin{prop}\label{prop:locdet_b}
Let $u'(a)<\beta v'(b)$. Then MSS $P>0$ is locally determinate if $\gamma_v(b+P)<1+b/P$, where $\gamma_v(z)\coloneqq -zv''(z)/v'(z)>0$ denotes the relative risk aversion of $v$.
\end{prop}

\begin{proof}
Setting $\xi=\eta=P>0$ in \eqref{eq:Phi_partial} and using the equilibrium condition \eqref{eq:eqcond}, at MSS we obtain
\begin{equation*}
    \Phi_\xi+\Phi_\eta=u''(a-P)P+\beta v''(b+P)P<0.
\end{equation*}
If $\Phi_\eta>0$, dividing both sides by $\Phi_\eta$, we obtain
\begin{equation*}
    0>\frac{\Phi_\xi}{\Phi_\eta}+1=-\phi'(P)+1\iff \phi'(P)>1,
\end{equation*}
implying local determinacy. Therefore it suffices to show $\Phi_\eta>0$. Dividing \eqref{eq:Phi_eta} by $\beta v'(b+\eta)>0$ and setting $\eta=P>0$, we obtain
\begin{equation*}
    \frac{\Phi_\eta}{\beta v'(b+P)}=1-\gamma_v(b+P)\frac{P}{b+P}.
\end{equation*}
Therefore $\Phi_\eta>0$ if $\gamma_v(b+P)<1+b/P$.
\end{proof}

Proposition \ref{prop:b>0} assumes $b>0$. We next consider the case $b=0$.

\begin{prop}[\citealp{Scheinkman1980}]\label{prop:b=0}
Let $b=0$. If $\liminf_{z\to 0}zv'(z)>0$, then NMSS is locally determinate. The assumption holds if $\gamma_v(z)\ge 1$ in a neighborhood of $z=0$.
\end{prop}

\begin{proof}
Since $v$ is continuously differentiable with $v'>0$, $z\mapsto zv'(z)$ is positive and continuous for $z>0$. Since $\liminf_{z\to 0}zv'(z)>0$, there exists a constant $c>0$ such that $zv'(z)\ge c$ for $z>0$. Take any monetary equilibrium. Setting $b=0$ in \eqref{eq:eqcond}, we obtain
\begin{equation*}
    u'(a-P_t)P_t=\beta v'(P_{t+1})P_{t+1}\ge 
    \beta c.
\end{equation*}
Since $u'(a-P)P\to 0$ as $P\to 0$, $P_t$ must be bounded away from zero. Therefore the equilibrium is asymptotically monetary and NMSS is locally determinate.

Finally, suppose $\gamma_v(z)\ge 1$ for $z\in (0,\delta]$. Then $-v''(z)/v'(z)\ge 1/z$, so
\begin{equation*}
    \log\frac{v'(z)}{v'(\delta)}=-\log \frac{v'(\delta)}{v'(z)}=\int_z^\delta -\frac{v''(z)}{v'(z)}\diff z\ge \int_z^\delta \frac{1}{z}\diff z=\log \frac{\delta}{z}.
\end{equation*}
Exponentiating both sides yields
\begin{equation*}
    \frac{v'(z)}{v'(\delta)}\ge \frac{\delta}{z}\iff zv'(z)\ge \delta v'(\delta)>0.\qedhere
\end{equation*}
\end{proof}

Propositions \ref{prop:b>0} and \ref{prop:b=0} show that the local determinacy of NMSS depends on whether $b>0$ or $b=0$. Although the case $b=0$ is only a single point and is non-generic, it is nevertheless economically relevant because it may be regarded as a situation where the old have no income and must save for retirement. In a more general model, \citet{Santos1990} shows the existence of monetary equilibria in OLG economies when agents have no endowments of some future good under some assumption on the curvature of the utility function.

\end{document}